\documentclass[12pt]{article}
\usepackage{subfigure}
\usepackage{epsfig}
\usepackage{ifthen}
\usepackage{amsfonts}
\usepackage{fullpage}
\usepackage{amsthm}
\usepackage{authblk}
\newtheorem{theorem}{Theorem}
\newtheorem{lemma}{Lemma}
\newtheorem{definition}{Definition}
\usepackage[bottom]{footmisc}

\long\def\cut#1{{}}
\def\diam{{\rm diam}}
\begin{document}
\newcommand{\path}[1]{
  \ifthenelse{\equal{#1}{}}{
    \Pi
  }{
    \Pi_{#1}
  }
}

\newcommand{\subpath}[3]{
  \ifthenelse{\equal{#1}{}}{
    \Pi(#2,#3)
  }{
    \Pi_{#1}(#2,#3)
  }
}
\newcommand{\sublen}[3]{
  \ifthenelse{\equal{#1}{}}{
    d_{\path{}}(#2,#3)
  }{
    d_{\path{#1}}(#2,#3)
  }
}
\newcommand{\dist}[2]{
    d(#1,#2)
}
\newcommand{\phase}[1]{phase #1}
\title{
%
%
Capturing an Evader in Polygonal Environments: \\
A Complete Information Game\thanks{
The research on this paper was supported in part by the National Science Foundation
Grant IIS-0904501. A preliminary version of this paper appeared at the 25th 
Conference on Artificial Intelligence.}}

\author{Kyle Klein}
\author{Subhash Suri}
\affil{
Department of Computer Science\\
UC Santa Barbara, CA 93106, USA\\
\{kyleklein, suri\}@cs.ucsb.edu 
}
\date{}
\maketitle

\begin{abstract}

Suppose an unpredictable evader is free to move around in a polygonal environment of 
arbitrary complexity that is under full camera surveillance. How many pursuers, each 
with the same maximum speed as the evader, are necessary and sufficient to guarantee 
a successful capture of the evader? The pursuers always know the evader's current 
position through the camera network, but need to physically reach the evader to 
capture it. We allow the evader the knowledge of the current positions of all the pursuers 
as well---this accords with the standard worst-case analysis model, 
but also models a practical situation where the evader has ``hacked'' into the 
surveillance system.
Our main result is to prove that \emph{three} pursuers are always \emph{sufficient 
and sometimes necessary} to \emph{capture} the evader. The bound is independent of 
the number of vertices or holes in the polygonal environment. The result should be 
contrasted with the \emph{incomplete} information pursuit-evasion where at least 
$\Omega (\sqrt{h} + \log n)$ pursuers are required~\cite{Guibas99} just for detecting 
the evader in an environment with $n$ vertices and $h$ holes.
\cut{
In fact, ours is one of the few worst-case bounds known for \emph{capturing} the evader, 
not just detecting it, under the visibility-based pursuit model.
}
\end{abstract}

\section{Introduction}

Pursuit-evasion games provide an elegant setting to study algorithmic and strategic
questions of exploration or monitoring by autonomous agents. Their mathematical
history can be traced back to at least 1930s when Rado posed the now-classical Lion-and-Man
problem: a lion and a man in a closed arena have equal maximum speeds; what tactics should
the lion employ to be sure of his meal? The problem was settled by Besicovitch who showed
that the man can escape regardless of the lion’s strategy. An important aspect of this
pursuit-evasion problem, and its solution, is the assumption of continuous time: each
player’s motion is a continuous function of time, which allows the lion to get arbitrarily
close to the man but never capture him. If, however, the players move in discrete time steps,
taking alternating turns but still in continuous space, the outcome is different, as first
conjectured by Gale \cite{Guy91} and proved by Sgall \cite{sgall_lion}.

A rich literature on pursuit-evasion problem has emerged since these initial investigations,
and the problems tend to fall in two broad categories: discrete space, where the pursuit
occurs on a graph, and continuous space, where the pursuit occurs in a geometric space.
Our focus in this paper is on the latter: visibility-based pursuit in a polygonal environment
in two dimensions. There exist simply-connected $n$-gons that may require $\Omega (\log n)$ 
pursuers in the worst-case to detect a single, arbitrarily fast moving evader, and
$O(\log n)$ pursuers also always suffice for all $n$ vertex simple polygons~\cite{Guibas99}.
When the polygon has $h$ holes, the number of necessary and sufficient pursuers turns 
out to be $O(\sqrt{h} + \log n)$~\cite{Guibas99}. However, these results hold only 
for detection of the evader, \emph{not for the capture.}

For capturing the evader, it is reasonable to assume that the pursuers and the evader all 
have the same maximum speed.  Under this assumption, it is shown by Isler et 
al.~\cite{Isler05} that two pursuers can capture the evader in a \emph{simply-connected}
polygon using a \emph{randomized} strategy whose expected search time is polynomial in $n$ and the 
diameter of the polygon.  When the polygon has holes, no non-trivial upper bound is
known for capturing the evader. For instance, we do not even know if $O(h)$ pursuers
are able to capture the evader. Because visibility-based pursuit allows \emph{unbounded} 
line-of-sight visibility regardless of the distance, it is unclear how to map a 
detection strategy to a capture strategy.

In this paper, we attempt to \emph{disentangle} these two orthogonal issues inherent in pursuit 
evasion: \emph{localization}, which is purely an informational problem, and \emph{capture}, 
which is a problem of planning physical moves. In particular, we ask how complex is the 
capture problem \emph{if the evader localization is available for free}? In other 
words, suppose the pursuers have complete information about the evader's current position, 
how much does it help them to capture the evader? Besides being a theoretically interesting 
question, the problem is also a reasonable model for many practical settings. Given the rapidly 
dropping cost of electronic surveillance and camera networks, it is now both technologically 
and economically feasible to have such monitoring capabilities. These technologies enable 
cheap and ubiquitous detection and localization, but in case of intrusion, a physical capture 
of the evader is still necessary.

\subsection{Our Contribution}

Our main result is that under such a complete information setting, \emph{three pursuers} are 
always sufficient to capture an equally fast evader in a polygonal environment with holes,
using a \emph{deterministic} strategy.
The bound is independent of the number of vertices $n$ or the holes of the polygon, although 
the capture time depends on both $n$ and the diameter of the polygon. Complementing this upper 
bound, we also show that there exists polygonal environments that require at least
three pursuers to capture the evader even with full information.

\subsection{Related Work} \label{sec:RelatedWork}

There is an enormous literature on pursuit evasion and related 
problems~\cite{aigner,dan-graph,cop1,robot-rabbit,Isler06,isler,lapaugh,Park01,Par76,SacRajLav04,Suzuki92}.
The research tends to fall into two distinct categories: 
geometry-based and graph-based. The former assumes a continuous model of space, typically a 
polygon, while the latter assumes a discrete graph model where agents move along edges. 
The graphs provide a very general setting but can suffer from two shortcomings: one, the 
generality leads to weak upper bounds and, two, they fail to model many restrictions imposed by
the geometry of physical world. Thus, for instance, determining the search
number of cop-number or a general graph remains a difficult open problem despite decades 
of research.

In visibility-based pursuit, a seminal paper~\cite{Guibas99} shows that $\Theta(\log n)$ 
pursuers are both necessary and sufficient in worst-case for a \emph{simply-connected} 
$n$-vertex polygon. Most of the existing work in polygon searching, however, is on
\emph{detection} and not capture. The only relevant result on capture is by 
Isler et al.~\cite{Isler05} showing that in \emph{polygons without holes} two pursuers 
can achieve both detection and capture. When the environment has holes, it is not even 
known how many pursuers are sufficient to capture an evader, even though a tight bound
of $\Theta( \sqrt{h} + \log n)$ for detection is known. In one important aspect,
polygon searching is fundamentally different from graph searching: 
\emph{re-contamination} is unavoidable in polygons, in general, while graphs 
can always be searched \emph{optimally} without re-contamination~\cite{Guibas99}.

Our work bears some resemblance to, and is inspired by, the result of Aigner 
and Fromme~\cite{aigner} on planar graphs, showing that graph searching on
planar graph requires 3 cops.  In that work, the graph is unweighted,
does not deal with Euclidean distances, and require players to
move to only neighboring nodes. Unlike the graph model, our search occurs in 
continuous Euclidean plane, and players can move to any position within distance one.
Thus, while our bounds are similar, the proof techniques and technical details
are quite different.

\section{The Problem Formulation} \label{sec:TheProblem}

We assume that an evader $e$ is free to move in a two-dimensional closed polygon $P$, 
which has $n$ vertices and $h$ holes. A set of pursuers, denoted $p_1, p_2, \ldots$, 
wish to capture the evader. All the players have the same
maximum speed, which we assume is normalized to 1. The bounds in our algorithm depend 
on the number of vertices $n$ and the diameter of the polygon, $\diam(P)$, which is 
the maximum distance between any two vertices of $P$ under the shortest path metric.

For the sake of notational brevity, we also use $e$ to denote the current position of 
the evader, and $p_i$ to denote the position of the $i$th pursuer. We model the 
pursuit-evasion as a continuous space, discrete time game: the players can move anywhere 
inside the polygon $P$, but they take turns in making their moves, with the
evader moving first.  In each move, a player 
can move to any position whose shortest path distance from its current position is at most 
one; that is, within \emph{geodesic disk} of radius one.  On the pursuers' move, all the 
pursuers can move simultaneously and independently.  We say that $e$ is successfully captured
when some pursuer $p_i$ becomes collocated with $e$.

In order to focus on the complexity of the capture, we assume a complete
information setup: each pursuer knows the location of the evader at all times.
We also endow the evader the same information, so $e$ also knows the locations of all
the pursuers. In addition, both sides know the environment $P$, but neither
side knows anything about the \emph{future} moves or strategies of the other side.
We begin with a high level description of the capture strategy, followed by 
its technical details and proof of correctness in the next section.

\subsection{The High Level Strategy for Capture}

We show that three pursuers, denoted $p_1, p_2, p_3$, can always capture an evader
using a deterministic strategy,
regardless of the evader's strategy and the geometry of the environment.
Our overall strategy is to progressively trap the evader in an ever-shrinking
region of the polygon $P$. The pursuit begins by first choosing a path $\path{1}$
that divides the polygon into sub-polygons (see Figure~\ref{fig:polygonPath})---we will use the notation $P_e$
to denote the sub-polygon containing the evader. We show that, after an 
initialization period, the pursuer $p_1$ can successfully guard the 
path $\path{1}$, meaning that $e$ cannot move across it without being captured.

In a general step, the sub-polygon $P_e$ containing the evader is bounded by two 
paths $\path{1}$ and $\path{2}$, satisfying a geometric property called \emph{minimality}, 
each being guarded by a pursuer. We then choose a third path $\path{3}$ splitting the 
region $P_e$ into two non-empty subsets.  If both regions have holes, then we argue 
that the pursuer $p_3$ can guard $\path{3}$, thereby trapping $e$ either between 
$\path{1}$ and $\path{3}$ (see Figure~\ref{fig:confine}), or between $\path{2}$ and 
$\path{3}$, in which case the pursuit iterates in a smaller region.  If one of the 
regions is hole-free, then we show that the pursuer $p_3$ can \emph{evict} the evader from 
this region, forcing it into a smaller region where the search resumes.

\subsection{Visibility Graphs and Path Guarding}

In order for this strategy to work, the paths $\path{i}$ need to be carefully chosen
and must satisfy certain geometric conditions, which we briefly explain.
First, although the pursuit occurs in continuous space, our paths will be computed
from a \emph{discrete} space, namely, the \emph{visibility graph} of the polygon.
The visibility graph $G(P)$ of a polygon $P$ is defined as follows:
	the nodes are the vertices of the polygonal environment
(including the holes), and two nodes are joined by an edge if the line segment
joining them lies entirely in the (closed) interior of the polygon.
(In other words, the two vertices joined by an edge must have line of sight visibility.)
This \emph{undirected} graph has $n$ vertices and at most $O(n^2)$ edges. 
We assign each edge a \emph{weight} equal to the Euclidean distance between 
its two endpoints. See Figure~\ref{fig:polygonPath} for an example.
\begin{figure}[htb]
\begin{center}

\subfigure[]{
\includegraphics[scale=.60]{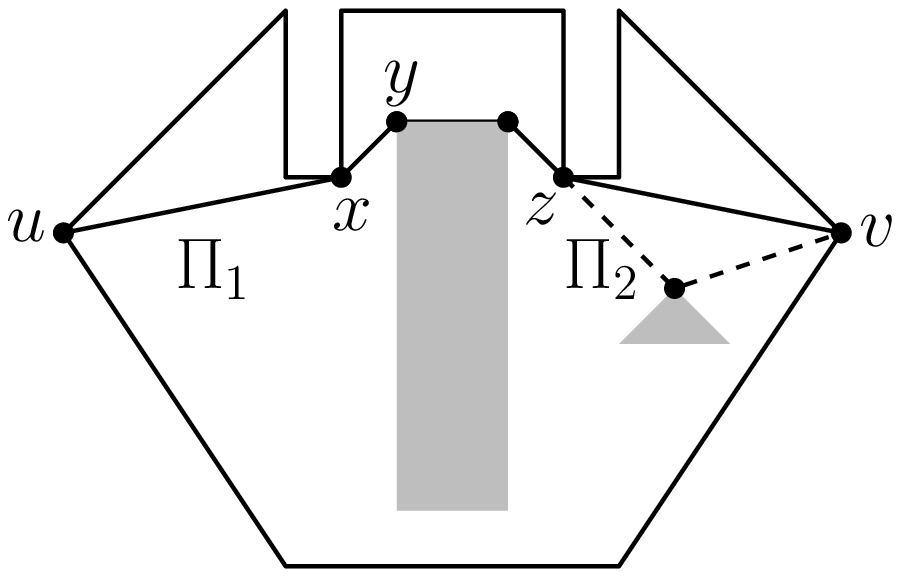}
\label{fig:polygonPath}
}	\hspace*{2cm}
\subfigure[]{
\includegraphics[scale=.60]{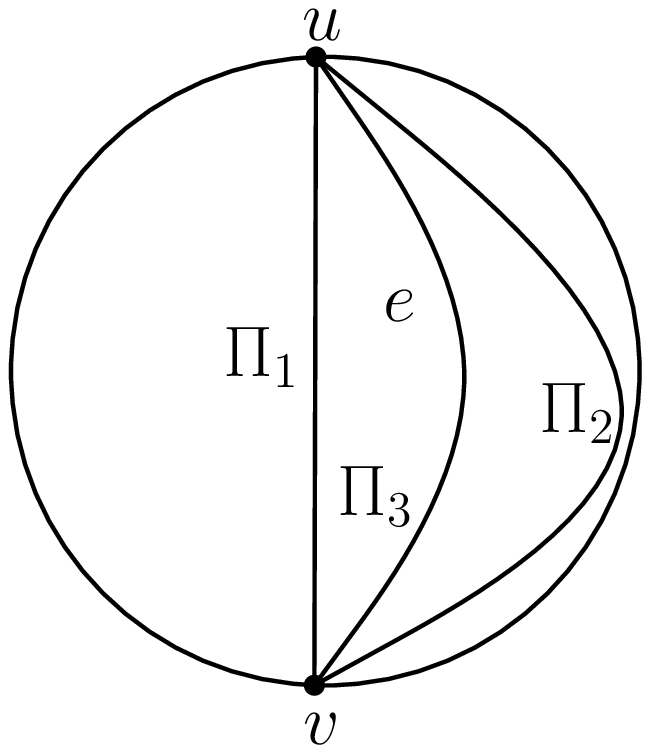}
\label{fig:confine}
}
\end{center}
\caption[toc entry]{(a) A polygonal environment with two holes
(a rectangle and a triangle). $xy$ is a visibility edge of $G(P)$, while
$xz$ is not. $\path{1}$ and $\path{2}$ are the first and the second
shortest paths between anchors $u$ and $v$.
The figure (b) illustrates the main strategy of trapping the evader
through three paths. }
\label{fig:mainStrategy}
\end{figure}

One can easily see that, given two vertices $u$ and $v$ of $P$, the \emph{shortest
path} from $u$ to $v$ in $G(P)$ is also the shortest Euclidean path constrained
to lie inside $P$. (The shortest Euclidean path has corners only at vertices of $G(P)$.)
However, we cannot make such a claim for the \emph{second}, or in general the 
$k$th, shortest path---one can create an infinitesimal ``bend'' in the shortest 
path $\path{1}$ to create another path that is arbitrarily close to the first shortest 
path but does not belong to $G(P)$.
Therefore, we will only consider paths that belong to $G(P)$ and are ``combinatorially 
distinct'' from $\path{1}$---that is, they differ in at least one visibility 
edge. However, even then the $k$th shortest path between two nodes can exhibit 
counter-intuitive behavior. For instance, while in graphs with non-negative weights
the first shortest path is always loop-free, the \emph{second}, or more 
generally $k$th, shortest path can have loops---this may happen if repeatedly 
looping around a small-weight cycle (to make the path distinct from others) 
is cheaper than taking a different but expensive edge~\cite{suri-focs}. 
Therefore, we will consider only shortest loop-free paths. One of our technical lemmas 
proves that these paths are also \emph{geometrically} non-self-intersecting. 
(This is obvious for the shortest path $\path{1}$ but not for subsequent paths.)
In addition, we argue that these paths also satisfy a key geometric property, 
called \emph{minimality}, which allows a pursuer to guard them against an evader.

\section{Proof of Sufficiency of 3 Pursuers}

We begin with the discussion of how a single pursuer can guard a path in $P$, trapping 
the evader on one side. We then discuss the technically more challenging case of guarding 
the second and the third paths.
In order to guarantee that a path in $P$ can be guarded, it must satisfy certain
geometric properties. We begin by introducing two key ideas: a \emph{minimal path} and
the \emph{projection} of evader on a path.
In the following, we use the notation $\dist{x}{y}$ to denote the shortest path
distance between points $x$ and $y$. When we require that distance to be measured within 
a subset, such as restricted to a path $\path{}$, we write $\sublen{}{x}{y}$.
That is, $\sublen{}{x}{y}$ is the length of path $\path{}$ between its points
$x$ and $y$.  Occasionally, we also use the notation  $\subpath{}{x}{y}$ to denote subpath of
$\path{}$ between points $x,y$. We use the notation $x \prec y$ to emphasize that the 
point $x$ precedes $y$ on the path $\path{}$: that is, if $\path{}$ is the
path from node $u$ to node $v$, then $x \prec y$ means that $\sublen{}{u}{x} < \sublen{}{u}{y}$.
The following property is important for patrolling of paths.

\begin{definition}
{({\bf Minimal Path:})} Suppose $\path{}$ is a path in $P$ dividing it into two
        sub-polygons, and $P_e$ is the sub-polygon containing the evader $e$. 
	We say that $\path{}$ is \emph{minimal} if, for all points $x,z \in \path{}$ and 
	$y \in (P_e \setminus \path{})$, the following holds:
	$$ \sublen{}{x}{z} ~\leq~ \dist{x}{y} + \dist{y}{z} $$
\end{definition}

Intuitively, a minimal path cannot be shortcut: that is, for any two points
on the path, it is never shorter to take a detour through an interior point 
of $P_e$. (This is a weak form of triangle inequality, which excludes
detours only through points contained in $P_e.)$
The next definition introduces the projection of the evader on to a path,
which is an important concept in our algorithm.

\begin{definition}
{(~{\bf Projection:})} 	Suppose $\path{}$ is a path in $P$ dividing it into two
	sub-polygons, and $P_e$ is the sub-polygon containing the evader $e$. Then, the 
	\emph{projection} of $e$ on $\path{}$, denoted $e_\pi$, is a point on $\path{}$ such 
	that, for all $x \in \path{}$, $e$ is no closer to $x$ than is $e_\pi$.
\end{definition}

Thus, if a pursuer is able to position itself at the projection of $e$ at all times,
then it guarantees that the evader cannot cross the path without being captured.
With these definitions in place, we now discuss how to guard the first path $\path{1}$.

\subsection{Guarding the First Path}

We choose two non-neighbor vertices $u$ and $v$ on the outer boundary of $P$, and call 
them \emph{anchors}. We let $\path{1}$ be the shortest path from $u$ to $v$ in $G(P)$;
this is also the shortest Euclidean path between $u$ and $v$ constrained to lie inside 
the environment. Our first observation is that this path $\path{1}$ is always
minimal.

\begin{lemma} 
\label{lem:minimal}
The path $\path{1}$ between $u$ and $v$ is minimal.
\end{lemma}

\begin{proof}
For the sake of contradiction, suppose there are two points $x,z \in \path{1}$
that violate the minimality. Let the point $y \notin \path{1}$ be the witness of
this violation.
That is, $\dist{x}{y} + \dist{y}{z} < \sublen{1}{x}{z}$. But then $\path{1}$ is
not the shortest path from $u$ to $v$, because its subpath $\subpath{1}{x}{z}$
is sub-optimal.         This completes the proof.  
\end{proof}

The following lemma shows that the projection of $e$ is well-defined for a minimal path.

\begin{lemma} \label{lem:pathEqualDistance} Suppose $\path{}$ is a
  minimal path between the anchor nodes $u$ and $v$.  Then, for every
  position of the evader $e$ in $P_e$, the projection $e_\pi$ exists.
  In fact, point on $\path{}$ at distance $\dist{e}{u}$ from $u$ along
  the path is a projection of $e$.
\end{lemma}
\begin{proof}
On the path $\path{}$, starting from $u$, we pick the furthest point $z$ such that for all 
$x \prec z$ we have $\sublen{}{x}{z} \leq \dist{x}{e}$. We note that necessarily for some 
$x\neq z$, this must be equality because otherwise there exists a point farther along $\path{}$
satisfying the inequality.
We claim that $z$ is a projection of $e$. Suppose not. Then there exists a point 
$x \succ z$ such that $\sublen{}{z}{x} > \dist{e}{x}$. This violates the minimality of
$\path{}$ because

\vspace*{-3ex}

\begin{eqnarray*}
\sublen{}{u}{x} & = & \sublen{}{u}{z} + \sublen{}{z}{x} = \dist{u}{e} + \sublen{}{z}{x} 
	 	~ > ~ \dist{u}{e} + \dist{e}{x}
\end{eqnarray*}

Next, we show that the point at $\dist{e}{u}$ along $\path{}$ is a
projection of $e$.  In case we have $\dist{e}{u}>\sublen{}{u}{v}$, we
choose $v$ as the projection point as a convention. Otherwise,
consider the point $z$ that lies on $\path{}$ at distance
$\dist{e}{u}$ from $u$. In order to show that $z$ is a valid and
unique projection, it suffices to show that for any other point $x$ on
$\path{}$, the point $z$ is closer to it than $e$ is.  Indeed, if $x$
is such that $z \prec x$, then $\sublen{}{z}{x} > \dist{e}{x}$ would
contradict the minimality of $\path{}$:

\vspace*{-3ex}

\begin{eqnarray*}
\sublen{}{u}{x} & = & \sublen{}{u}{z} + \sublen{}{z}{x} ~=~
\dist{u}{e} +  \sublen{}{z}{x} ~>~ \dist{u}{e} +  \dist{e}{x}
\end{eqnarray*}

Similarly, if  $x$ is such that $x \prec z$, then $\sublen{}{x}{z} > \dist{x}{e}$
contradicts the hypothesis that $y$ satisfies $\dist{e}{u} = \sublen{}{u}{z}$:

\vspace*{-3ex}

\begin{eqnarray*}
\dist{e}{u} & \leq & \dist{e}{x} + \sublen{}{x}{u} ~<~ \sublen{}{z}{x} + \sublen{}{x}{u} ~=~ 
\sublen{}{z}{u}
\end{eqnarray*}

Therefore, the chosen point $z$ is closer to all points on $\path{}$
than $e$ is, and hence it is a projection of $e$.  This completes the
proof.
\end{proof}

The next lemma shows how a pursuer can guard a minimal path.

\begin{lemma} \label{lem:singlepath}
Suppose $\path{}$ is a minimal path between the anchors $u,v$ in $P$, and a pursuer $p$ 
is located at the current projection of $e$. 
Suppose on its turn the evader moves from $e$ to $e'$.
Then, the pursuer $p$ can either capture the evader or relocate to the new projection 
$e'_{\pi}$ in one move.
\end{lemma}
\begin{proof}
First, suppose that the new position $e'$ is on different side of the path $\path{}$ than $e$;
that is, the evader crosses the path. Because $e$ moves a distance of at most one, and suppose 
it crosses the path at a point $z$, we have $\dist{e}{z} + \dist{z}{e'} \leq 1$.
On the other hand, since $p$ is located at the projection of $e$ before the move,
$\sublen{}{p}{z} \leq \dist{e}{z}$. Therefore, the new position of the evader $e'$ is within 
distance one of $p$, and the pursuer can capture the evader on its move.

Therefore, assume that the evader does not cross the path, and moves to a position $e'$
such that $e_\pi \neq e'_\pi$.  Consider any two points $x \prec e_\pi \prec y$
(one on either side of the projection), and let 
$\dist{x}{e'} = \sublen{}{x}{e_\pi} + c$ and $\dist{y}{e'} = \sublen{}{y}{e_{\pi}} + c'$. 
Then we claim that $c + c' \geq 0$. We observe that by the minimality of $\path{}$,
the following holds, which in turn implies that $c+c'\geq 0$.

\vspace*{-2ex}

\begin{eqnarray*}
\sublen{}{x}{y} & \:\leq\: & \dist{x}{e'} + \dist{e'}{y} \:=\: \sublen{}{x}{e_{\pi}} + 
		\sublen{}{y}{e_{\pi}} +c + c'  \\
		& \:=\: & \sublen{}{x}{y} +c +c'
\end{eqnarray*}

\vspace*{-1ex}

Thus $e$ can move closer to $x$ or $y$ but not both.
Suppose it moves closer to $x$, meaning $c < 0$. 
Consider some $x \prec e_{\pi}$ with the smallest value of $c$.
Then suppose for any $y \succ e_{\pi}$, $c' < |c|$. 
This would imply $c + c' < 0$, and so if $e$ moves a maximum of $c$ closer to any $x$, 
it moves at least $c$ farther from any $y$. We claim that $e'_{\pi}$ is at the point 
$c$ closer to $x$. To show this we argue that at the position $e'_{\pi}$, the pursuer 
is closer to all points on $\path{}$. But $e'$ cannot be closer to any $x \prec e'_{\pi}$
because $e'$ is no more than $c$ closer than $e$ to any such $x$. Similarly $e'$ must still be farther from
any $y \succ p_e$ because $e$ moved $c$ or more away from those points. 
This leaves the points between $e'_{\pi}$ and $e_{\pi}$, but if $e'$ is closer to one of 
these points (say, $z'$), then $p$ can capture $e$ because $\dist{e}{z'} < \dist{e'_{\pi}}{z'}$.
Therefore, the pursuer can move to the new projection because $|c| < 1$.
This completes the proof.   
\end{proof}

Finally, we show that a pursuer $p$ requires $O(\diam(P)^2)$ moves to
either reach the current projection of $e$ or capture it. 

\begin{figure}[htb]
\begin{center}
\includegraphics[scale=.6]{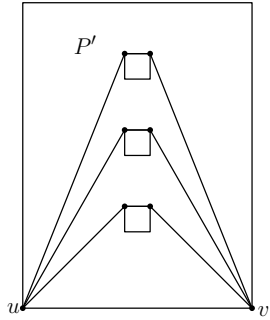}
\end{center}
\caption[toc entry]{Example of sub-polygons $P'$ created with diameter
  larger than $\diam(P)$.} 
\label{fig:growingPaths}
\end{figure}

\begin{lemma} \label{lem:initialize}
Suppose $\path{}$ is a minimal path between anchors $u,v$ in $P$, and a pursuer $p$ is 
located at $u$.  Then in $O(\diam(P)^2)$ moves, $p$ can move to $e$'s projection.
\end{lemma}
\begin{proof}
  By Lemma~\ref{lem:singlepath}, the projection of $e$ can only shift
  by distance at most one along the path $\path{}$. Thus, $p$'s
  strategy is simply to move along the path from one end to the other
  until it coincides with the current projection of $e$, or captures
  it. However, as $\path{}$ is not the shortest $u,v$ path,
  $\sublen{}{u}{v}$ may be larger than $\diam(P)$ (an example of
  progressively longer paths is depicted in Figure~\ref{fig:growingPaths}), thus we
  show that $\sublen{}{u}{v} \leq \diam(P)^2$.

  Notice that $\path{}$ splits $P$ into two or more sub-polygons, and
  so let us consider one such sub-polygon $P$'. If $\diam(P') \leq
  area(P')$, then trivially as $area(P') < area(P)$ we can deduce
  $\sublen{}{u}{v} \leq area(P)$ and necessarily $\sublen{}{u}{v} \leq
  \diam(P)^2$. However, suppose that $\diam(P') > area(P')$, we can
  easily reverse this inequality by a simple \emph{rescaling of the
    units}.  In particular, suppose we rescale the unit of measurement
  from $1$ to $1+\alpha$.  This increases the area of a triangle by a
  factor of $(1+\alpha)^2$, while a segment only increases in length
  by a factor of $1+\alpha$. A suitably large choice of $\alpha$,
  therefore, always ensures that the polygon's area exceeds its
  diameter, because the former grows by a factor of $(1+\alpha)^2$
  while the latter grows only by a factor of $1+\alpha$.  Thus, we
  assume that a suitable rescaling has been applied to the polygon
  $P$, ensuring that all sub-polygons $P'$ encountered during the
  algorithm satisfy $diam(P') \leq area(P')$.

  Since $p$ moves a distance of $1$ in each turn, and the path
  $\path{}$ is at most $\diam(P)^2$ long, the entire initialization
  phase takes $O(\diam(P)^2)$ time.  Meanwhile, if the projection ever
  ``crosses over'' the current position of $p$, the pursuer
  immediately can move to the new projection point because at that
  moment it must be within distance one of the target location.  This
  completes the proof.
\end{proof}

\subsection{Geometric Structure of Pursuer Paths}

We now come to the main part of our pursuit strategy. The key idea is to 
progressively trap the evader in a region bounded by two minimal paths, 
which are guarded by two pursuers, and to use the third pursuer to further 
divide the current trap region.  When the third pursuer subdivides the current 
region containing $e$, two possibilities emerge: either both regions of the 
subdivision have holes, in which case we show that the third path is necessarily
minimal and thus guardable by the third pursuer, limiting the evader to 
a smaller region than before; or one of the regions is hole-free, in which 
case the third pursuer uses the capture strategy for a simply-connected polygon to evict 
the pursuer from this region (or capture it). In order to formalize our 
strategy, we first show a key geometric property of the second and third 
shortest paths between the anchors in the visibility graph, namely, that they are 
non-self-intersecting, and therefore lead to well-defined closed regions.

\begin{figure}[htb]
\begin{center}
\includegraphics[scale=.6]{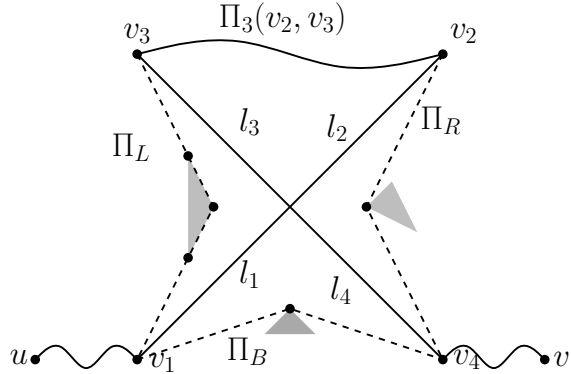}
\end{center}
\caption[toc entry]{Non-self-crossing of shortest paths $\path{1}, \path{2}, \path{3}$.} 
\label{fig:crossings}
\end{figure}

\begin{lemma} \label{lem:crossings}
Let $\path{1}$ be the shortest path between two anchor points $u$ and $v$ on $P$'s 
boundary, and focus on the sub-polygon $P_e$ that lies on one side of $\path{1}$.
Let $\path{2}$ and $\path{3}$, respectively, be the second and the third simple
(loop-free) shortest paths in the visibility graph $G(P_e)$ between $u$ and $v$.
Then, $\path{2}$ and $\path{3}$ are non-self-crossing.
\end{lemma}
\begin{proof}
Without loss of generality, suppose the path $\path{3}$ violates the lemma, and that 
two of its edges $(v_1, v_2)$ and $(v_3, v_4)$ intersect.  See Figure~\ref{fig:crossings}. 
We first note that the intersection point cannot be a vertex of the visibility graph
because otherwise the path has a cycle, and we assumed that $\path{3}$ is loop-free.
As shown in the figure, we break the segment $(v_1, v_2)$ into $l_1$ and $l_2$,
and $(v_3,v_4)$ into $l_3$ and $l_4$. By the triangle inequality of the Euclidean
metric, it is easy to see that $\dist{v_1}{v_3} < l_1 + l_3$, and
$\dist{v_2}{v_4} < l_2 + l_4$. Similarly, $\dist{v_1}{v_4} < l_1 + l_4$. 
Let $\path{L}$, $\path{R}$, $\path{B}$, respectively, denote the
shortest paths (in the graph $G(P_e))$ realizing these distances.
Now consider the following three paths between $v_1$ and $v_4$, each contained
in $G(P_e))$ and \emph{non-self-intersecting}:
$\path{L} \cdot \subpath{3}{v_2}{v_3} \cdot \path{R}$,
$\path{B}$, and
the shorter of
$\path{L} \cdot (v_3, v_4)$ and $(v_1, v_2) \cdot \path{R}$.
They are all shorter than $\path{3}$ and at least one of them
must be distinct from both $\path{1}$ and $\path{2}$, thus
contradicting the choice of $\path{3}$.
This completes the proof.   
\end{proof}


\subsection{Shrinking, Guarding and Evicting}

In a general step of the algorithm, assume that the evader lies in a region 
$P_e$ of the polygon bounded by two minimal paths $\path{1}$ and $\path{2}$ 
between two anchor vertices $u$ and $v$.  (Strictly speaking, the region
$P_e$ is initially bounded by $\path{1}$, which is minimal, and portion of $P$'s
boundary, which is not technically a minimal path. However, the evader cannot 
cross the polygon boundary, and so we treat this as a special case of the 
minimal path to avoid duplicating our proof argument.)
We also assume that $\path{1}$ and $\path{2}$ only share vertices $u$ and $v$;
if they share a common prefix or suffix subpath, we can delete those
and advance the anchor nodes to the last common prefix vertex and
the first common suffix vertex. This ensures that the region $P_e$
is \emph{non-degenerate}.
Furthermore, we assume that the region $P_e$ contains at least one 
hole---otherwise, the evader is trapped in a simply-connected region,
where a single (the third) pursuer can capture it.

The key idea of our proof is to show that, in the visibility graph $G(P_e)$,
if we compute a \emph{shortest path} from $u$ to $v$ that is distinct from 
both $\path{1}$ and $\path{2}$, then it divides $P_e$ into \emph{only} two 
regions, and that the evader is trapped in one of those regions.
We will call this new path the \emph{third} shortest path $\path{3}$.
Specifically, $\path{3}$ is the simple (loop-free) shortest path
from $u$ to $v$ in $G(P_e)$ distinct from $\path{1}$ and $\path{2}$.
(One can compute such a path using any of the algorithms for computing 
$k$ loop-free shortest paths in a weighted undirected graph~\cite{suri-focs,widmayer,Yen71}.)

\begin{lemma} \label{lem:twoRegions}
The shortest path $\path{3}$ between the anchor nodes $u$ and $v$
divides the current evader region $P_e$ into two connected regions.
\end{lemma}
\begin{proof}
If the path is disjoint from $\path{1}$ and $\path{2}$ except at endpoints, 
then $P_e$ is clearly subdivided into two regions. If $\path{3}$ shares
vertices only with $\path{1}$ or only with $\path{2}$, \emph{but in multiple disjoint
subpaths creating multiple regions}, then minimality of those paths means that we 
can contract all but one region and shorten $\path{3}$, contradicting the choice of
this third path. Therefore, let us suppose that $\path{3}$ shares vertices with 
both the paths, and so ``hops'' between $\path{1}$ and $\path{2}$, sharing 
common subpaths with them, and creates three or more regions.
In that case, $\path{3}$ must leave and rejoin $\path{1}$ and $\path{2}$ at 
least once, as shown by points $x,y,z$ in Figure~\ref{fig:twoRegions}. 
We observe that $\sublen{2}{y}{v}$ is no longer than 
$\dist{y}{z} + \sublen{1}{z}{v}$,
otherwise we can contradict the choice of $\path{2}$. 
Thus the third region can be removed by altering $\path{3}$ to use the subpath
$\subpath{2}{y}{v}$. 
(A symmetric case arises when the roles of $\path{1}$ and $\path{2}$ are swapped.)
Thus, we conclude that $\path{3}$ can create only two subregions.   
\end{proof}

\begin{figure}[htb]
\begin{center}
\subfigure[]{
\includegraphics[scale=.60]{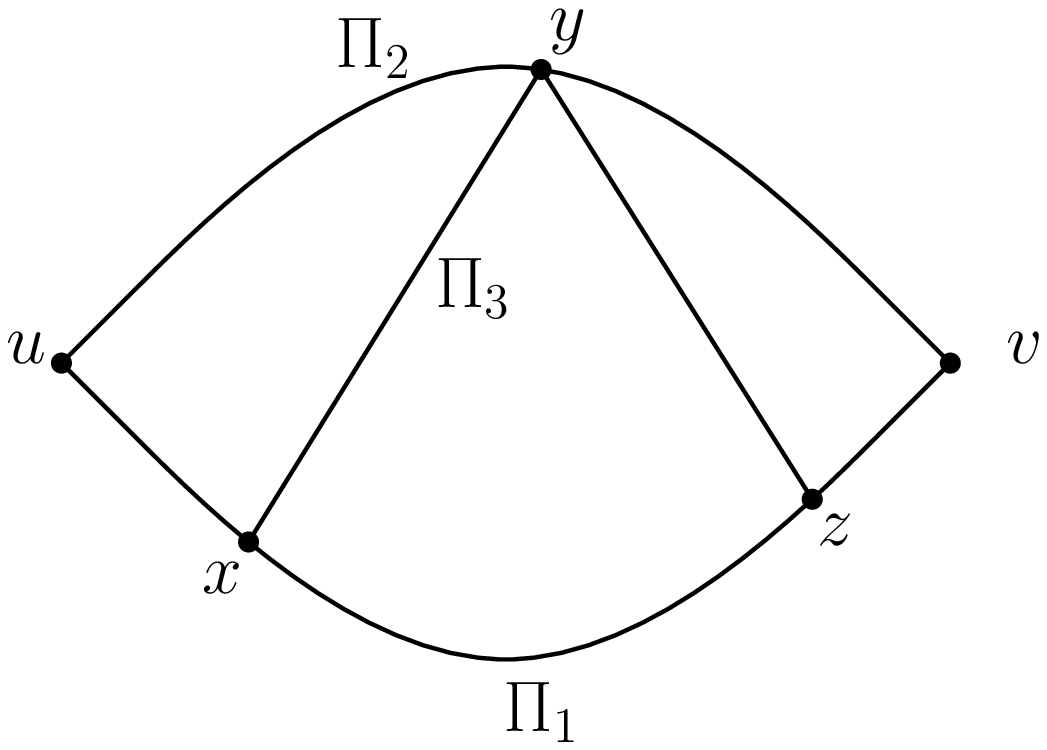}
\label{fig:twoRegions}
}
\subfigure[]{
\includegraphics[scale=.60]{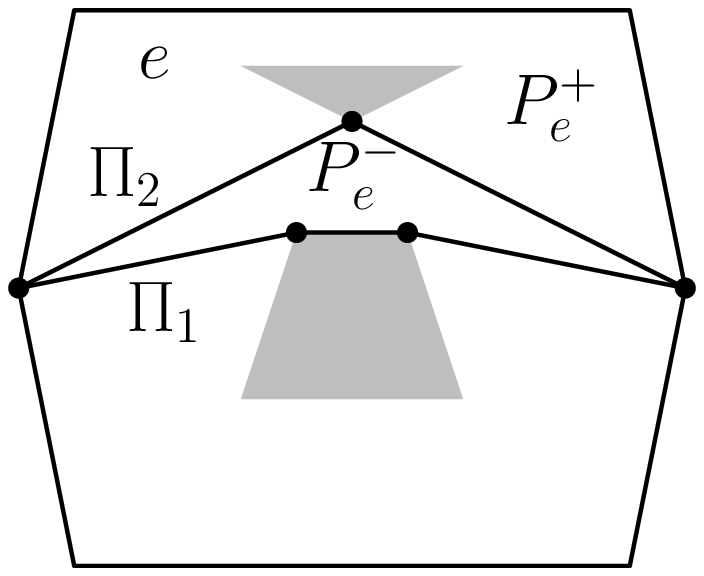}
\label{fig:polygonTwoRegion}
}
\end{center}
\caption[toc entry]{The left figure illustrates the proof of Lemma~\ref{lem:twoRegions};
the right figure illustrates the two subregions created by a path, $\path{2}$
in this case.}
\label{fig:shorterST}
\end{figure}

Because $P_e$ contains at least one hole, one of the regions created by 
the third shortest path $\path{3}$ also must contain a hole. The following 
lemma argues that $\path{3}$ is minimal with respect to that region.
(The other region may be hole-free, and we will argue that the evader
can be evicted from such a hole-free region or captured.)

\begin{figure}[htb]
\begin{center}
 \subfigure[]{
 \includegraphics[scale=.6]{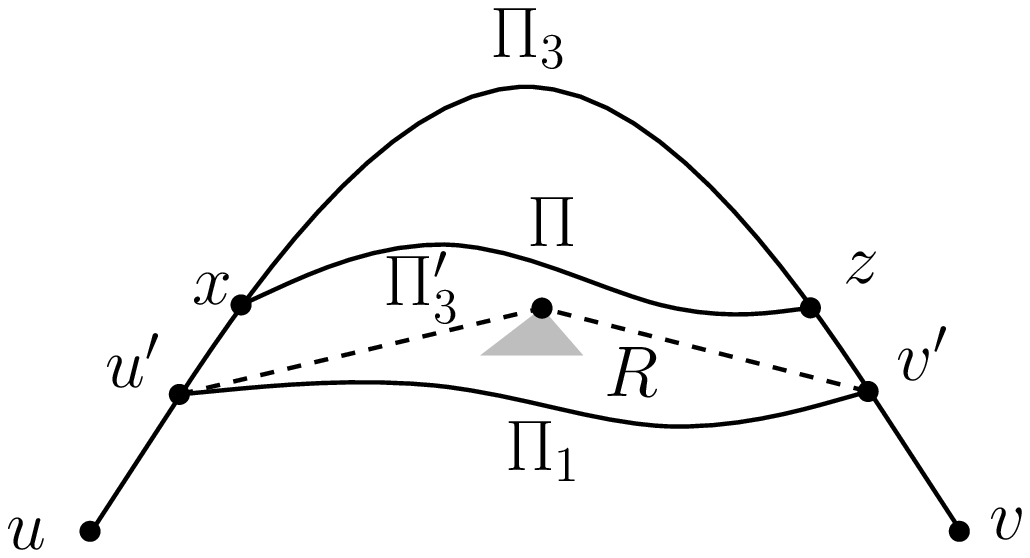}
 \label{fig:holeBelow}
 }
 \subfigure[]{
 \includegraphics[scale=.6]{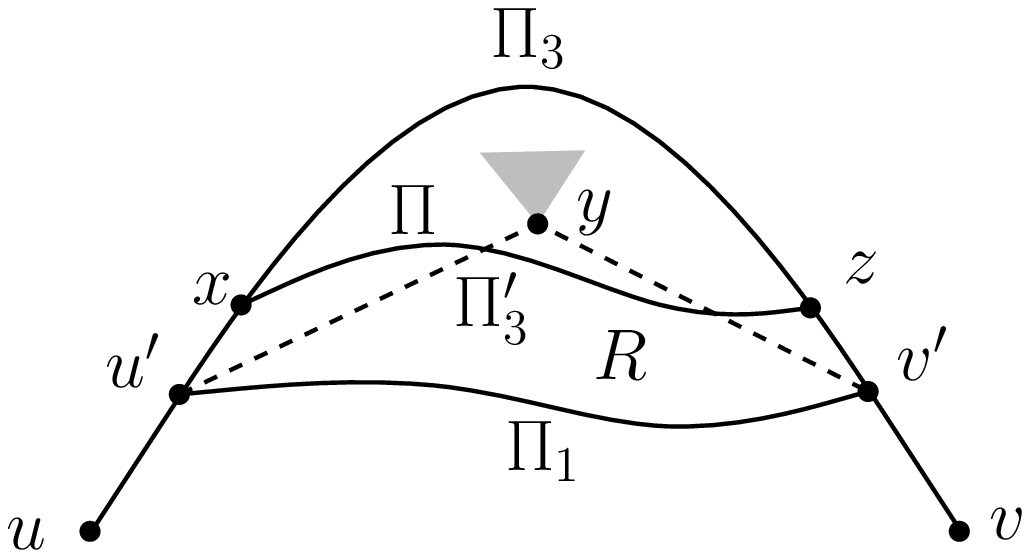}
 \label{fig:holeAbove}
 }
\end{center}
\caption[toc entry]{The left figure illustrates the proof of Lemma~\ref{lem:twoRegions};
the right figure illustrates the two subregions created by a path, $\path{2}$
in this case.}
\label{fig:holeAboveBelow}
\end{figure}

\begin{lemma} \label{lem:confinementPath}
Suppose $\path{3}$ divides the region $P_e$ into two subregions $P^+_e$ and $P^-_e$, 
and assume that $P^+_e$ contains at least one hole. Then, $\path{3}$ is a minimal path 
within the region $P^+_e$.
\end{lemma}
\begin{proof}
Assume, for the sake of contradiction, that the minimality of $\path{3}$ is violated for
two points $x,z \in \path{3}$. Let $u'$ be the vertex immediately preceding the point $x$,
possibly $x = u'$, and $v'$ is the vertex immediately following $z$, possibly $z = v'$, on
$\path{3}$. Consider the shortest path in $G(P_e)$ from $u'$ to $v'$. If this path is
\emph{not} a subpath of either $\path{1}$ or $\path{2}$, then we can immediately improve the
length of $\path{3}$ by using this subpath, thereby contradicting the choice of $\path{3}$.
Therefore, assume without loss of generality that the shortest path from $u'$ to $v'$
is a subpath of $\path{1}$. Further, let $\path{}$ denote the shortest path from
point $x$ to point $z$ in $P^+_e$, and consider the region $R$ bounded by 
$\subpath{1}{u'}{v'}$, $\path{}$ and the segments $(z,v')$ and $(x,u')$. If there are any holes in $R$  
then there is a distinct path $\path{3}'$ shorter than $\path{3}$ obtained by 
tightening $\path{}$ around those holes as shown in Figure~\ref{fig:holeBelow}. Thus 
the hole in $P^+_e$ must be outside $R$, however pick the closest vertex on a 
hole in $P^+_e$ to $\path{}$, call it $y$. Then a path $\path{3}'$ shorter than 
$\path{3}$ can be obtained using $y$ as shown in Figure~\ref{fig:holeAbove}.
Thus in all cases, if $P^+_e$ contains a hole, $\path{3}$ can be shortened, which
contradicts its optimality. Thus $\path{3}$'s minimality cannot be violated,
and the proof is complete.
\end{proof}

When one of the regions created by $\path{3}$ is hole-free, then
$\path{3}$ has a very simple structure, consisting of only two distinct edges
as seen in Figure~\ref{fig:polygonTwoRegion}, allowing it to be
cleared using the search strategy of a simply-connected polygon.

\begin{lemma} \label{lem:twoEdges}
Suppose $\path{3}$ divides the region $P_e$ into two subregions, one
of which is hole-free. Then, $\path{3}$ consists of precisely two edges.
\end{lemma}

\begin{proof}
Arguing as in the preceding lemma, it can be shown that if $P^+_e$ is
hole-free, then the shortest path $\path{3}$ has a common prefix and a
common suffix with $\path{1}$, and only differs in a subpath of the form
$\subpath{3}{u'}{v'}$. Suppose for the sake of contradiction that $\subpath{3}{u'}{v'}$
consists of more than two edges. Consider the first three points
of this subpath, $u'$, $x$ and $z$. Notice that because they are not
co-linear and by assumption there are no holes contained between $\subpath{3}{u'}{v'}$
and $\path{1}$ that the shortest $u',z$ path would shortcut $x$, contradicting
the choice of $\path{3}$. Thus $z=v'$ and $\path{3}$ consists of only two edges.  
This completes the proof.
\end{proof}

Now, if both regions created by $\path{3}$ have holes, then the minimality
of $\path{3}$ allows a third pursuer to guard this path, and the pursuit 
continues in this smaller region. If one of the regions created is hole-free, 
then we no longer can assume minimality with respect to that region, so a 
different strategy is required. The following lemmas show how to either 
capture the evader in such a region, or to force the evader out of 
(evict) this region, while guarding $\path{3}$ so the evader \emph{cannot 
reenter} this region.

Capturing an evader in a hole-free region can be accomplished by advancing 
along the shortest path towards the current position of the evader.
In particular, we can fix a origin $O$ in the region (say, some vertex
in $P$), and then letting the pursuer move along the
shortest path between $O$ and the current evader position. It can be
shown that the pursuer makes sufficient progress towards the evader,
as articulated in the following lemma (which paraphrases a
technical result of~\cite{isler}).

\begin{lemma}
\label{lem:isler}
After each move, (1) the pursuer $p$ remains on the shortest path between $O$ and $e$,
and (2) its new position $p'$ satisfies $\dist{p'}{O}^2 \geq \dist{p}{O}^2 + \frac{1}{n}$.
\end{lemma}

Using this result, we can derive the following lemma about capturing the
evader in a hole-free region. 

\begin{figure}[htb]
\begin{center}
\includegraphics[scale=.4]{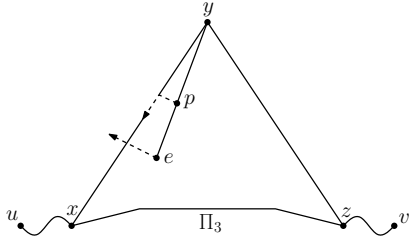}
\end{center}
\caption{An illustration of the pursuer's eviction strategy. Dashed lines denote moves where $e$ moved first.}
\label{fig:evict}
\end{figure}

\begin{lemma} \label{lem:removeEvader}
Suppose the evader lies in hole-free region of $k$ vertices that is bounded by $\path{3}$ and 
another minimal path. Then, in $O(k\cdot \diam(P)^2)$ moves, a single pursuer 
$p$ can either capture the evader or force it out of the region and
place itself on $e$'s projection on the path $\path{3}$.
\end{lemma}
\begin{proof}

Assume, without loss of generality, that our hole-free region is bounded by a minimal
path $\path{1}$ and the path $\path{3}$, which by Lemma~\ref{lem:twoEdges} must consist of two
edges, say, $(x,y)$ and $(y,z$). The pursuer $p$'s strategy is to move to $y$, and
execute a simple-polygon search with $y$ as the origin with the following modification: 
if $p$'s move takes it outside the region, then it moves along $\path{3}$ toward $e_{\pi}$ 
until $e$ reenters, at which point its resumes the pursuit.

As the shortest path between any two vertices consists of at most two edges, this region can
have diameter no larger than $2\cdot \diam(P)$. Thus if $e$ never leaves the region, then by
the known result of Lemma~\ref{lem:isler}, a successful capture occurs in $O(k\cdot \diam(P)^2)$ 
moves. Therefore, assume that $e$ leaves the region at some point. Since $\path{1}$ is minimal,
the evader cannot leave the region through that path, and so assume without loss of generality
that the evader crosses the segment $(x,y)$ of $\path{3}$. Because $p$ always stays on the shortest
path between $e$ and $y$, in an unmodified pursuit $p$'s move would cross $(x,y)$ as well. In the modified
pursuit, $p$ stops at the point where it crosses $(x,y)$ and advances toward the projection of 
$e$. See Figure~\ref{fig:evict} for illustration.

We note that the projection of $e$ is within distance one of where it crossed $(x,y)$. 
As a result, because $p$ crossed $(x,y)$ at a point closer to $y$ than $e$, if $e_{\pi}$ 
lies on the subpath $\subpath{3}{p}{z}$, then $p$ can reach it in one move. Otherwise, $p$ 
need simply advance forward along $\path{3}$ toward $x$. If $e$ never re-enters the 
hole free region, then by Lemma~\ref{lem:initialize} $p$ will reach the projection 
within $O(\diam(P)^2)$ moves.

In case $e$ re-enters the hole-free region, we note that it must do so by crossing the 
segment $(x,p)$, and that for each turn $e$ was outside the hole-free region $p$ moved 
distance one along the shortest path from $y$ to $e$.  Thus on its next turn $p$ can resume 
its pursuit, while having sufficiently increased its distance from $y$ to guarantee a 
successful capture occurs in $O(k\cdot \diam(P)^2)$ moves should $e$ remain within the 
hole-free region. Thus $e$ may continually move back and forth between the hole-free 
region, but within $O(k\cdot \diam(P)^2)$ moves $e$ will either be captured, or 
the pursuer will successfully guard $\path{3}$ by reaching the projection.
This completes the proof.
\end{proof}

We can now summarize our main result.

\begin{theorem}	\label{thm:main}
Three pursuers are always sufficient to capture an evader in $O(n\cdot \diam(P) ^2 )$ moves
in a polygon with $n$ vertices and any number of holes.
\end{theorem}
\begin{proof}

Whenever a new path is introduced, the size (number of vertices) of the region $P_e$ containing
$e$ shrinks by at least one. Thus, the number of different paths guarded during the
course of the pursuit before $e$ is trapped in a hole-free region is at most $n$.
Guarding each path requires $O(\diam(P)^2)$ moves for a minimal path, and
$O(k \cdot \diam(P)^2)$ moves when in a hole-free region with $k$ vertices.
Since the evader cannot reenter hole-free regions once they have been
guarded, the total cost of guarding all the hole-free regions during the
course of the algorithm sums to $O(n \cdot \diam(P)^2)$.  

Finally $e$ will be confined between two minimal paths in a hole-free
region consisting of three vertices, otherwise additional $u,v$ paths
can be found to further reduce the region. This sub-polygon clearly
has diameter no larger than $\diam(P)$, and thus the evader can be
captured in $O(\diam(P)^2)$ moves with the known result of
Lemma~\ref{lem:isler}, for a total of $O(n \cdot \diam(P)^2)$ moves
over the entire pursuit.
\end{proof}

\section{Necessity of 3 Pursuers} \label{sec:lower}

We show that any \emph{deterministic} strategy requires at least 3 pursuers
in the worst-case, and thus the upper bound of the previous section is tight.

\begin{theorem}
There exists an infinite family of polygons with holes that require at 
least three pursuers to capture an evader even with complete information 
about the evader's location.
\end{theorem}
\begin{proof}
The proof is based on a reduction from searching in \emph{planar graphs}.
In particular, consider a planar graph $G$, with minimum degree $3$, and 
without any cycles of length three or four (see Figure~\ref{fig:dodecahedron}). 
Using Fary's Theorem, we can  embed such a graph so that each edge maps 
to a straight line segment. By suitable scaling, assume that the longest 
edge in the embedding has length 1. (See Figure~\ref{fig:embeddedGraph} for 
an example.)

We now transform this straight-line embedding into a polygon with holes, by converting 
each edge into a ``corridor.'' Each corridor is constructed to ensure that the shortest 
path through it has length $1$. In particular, the edges of length $1$ map to straight 
corridors, while shorter edges correspond to corridors with multiple turns, as shown in 
Figure~\ref{fig:planarConstruction}~\subref{fig:edgeConstruction}.  It is easy to see 
that such a construction can ensure that all the corridors are non-overlapping.
With this transformation, the outer face of the graph becomes the boundary of the 
polygon $P$, while each face of the plane graph becomes a hole.

\begin{figure}[htb]
\begin{center}
\subfigure[]{
\includegraphics[scale=.6]{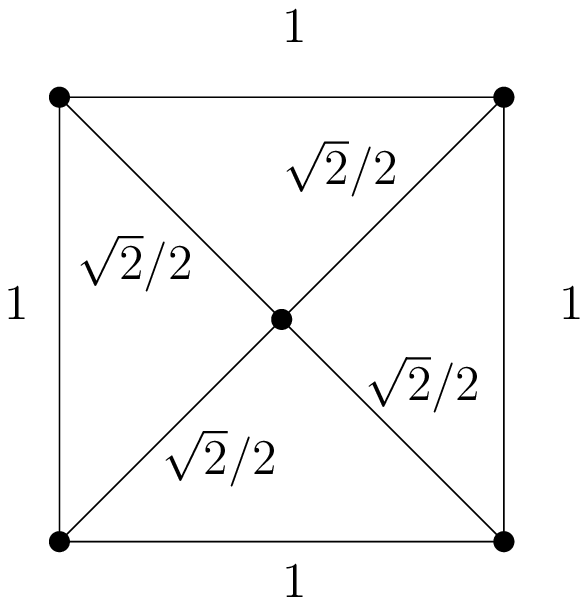}
\label{fig:embeddedGraph}
}
\subfigure[]{
\includegraphics[scale=.6]{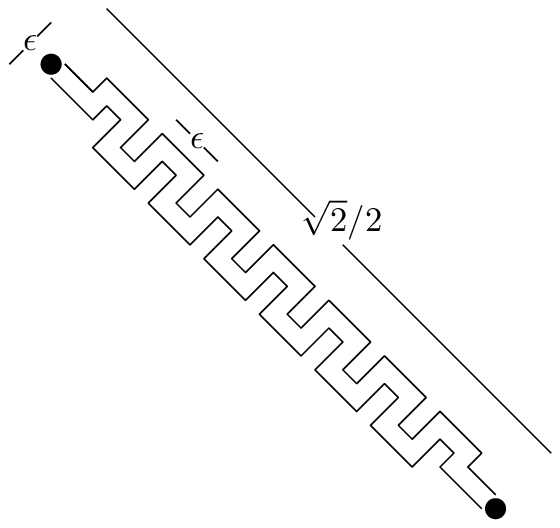}
\label{fig:edgeConstruction}
}
\subfigure[]{
\includegraphics[scale=.34]{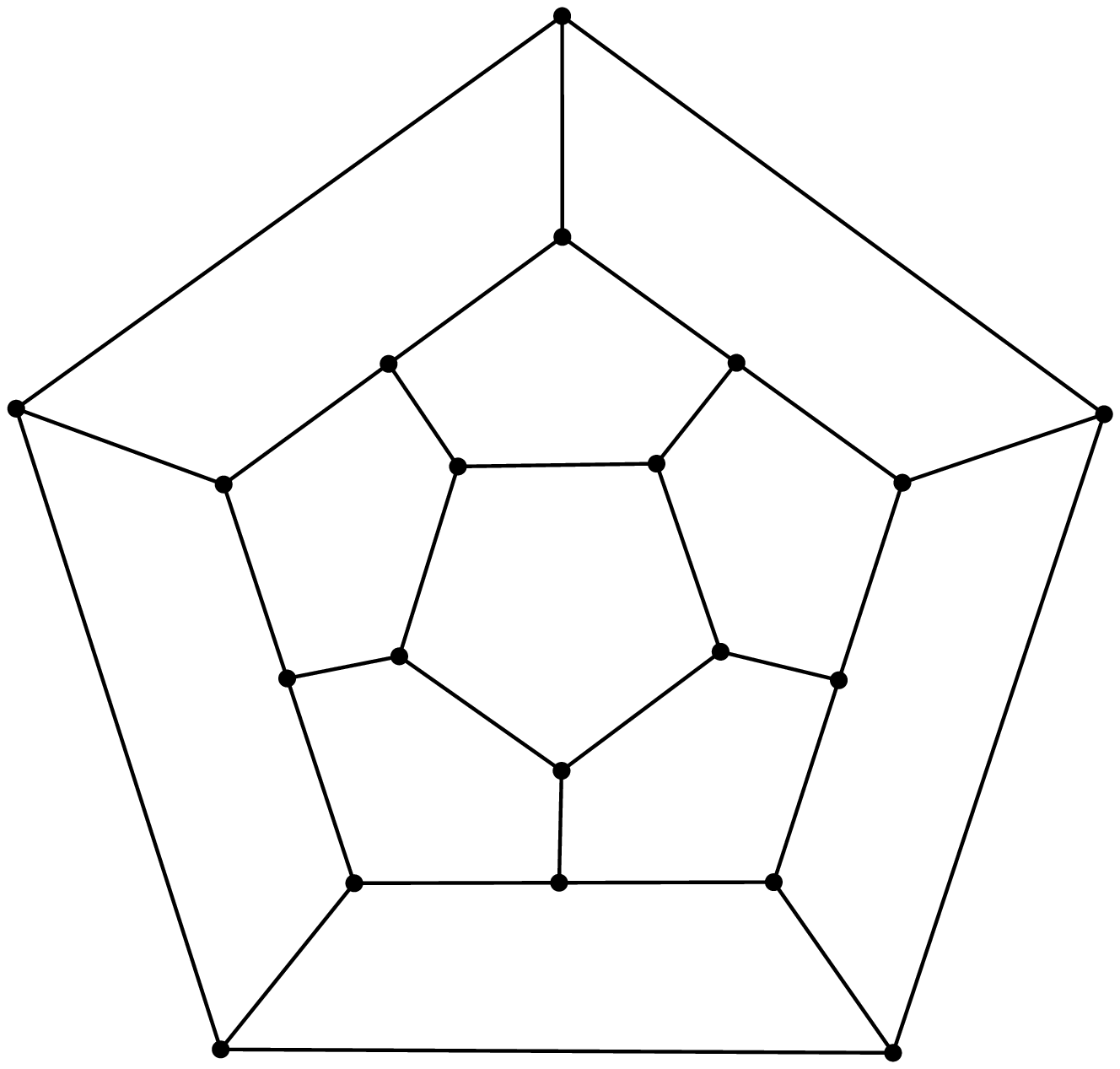}
\label{fig:dodecahedron}
}
\end{center}
\caption[toc entry]{Embedding of a planar graph \subref{fig:embeddedGraph}, 
corridor construction \subref{fig:edgeConstruction}, and
a planar graph with min-degree 3 and no three or four cycles \subref{fig:dodecahedron}.}
\label{fig:planarConstruction}
\end{figure}

It is known that in any graph with minimum degree $k$ and no cycles of 
length three or four, the evader has a winning strategy against $k-1$ 
pursuers~\cite{aigner}, as follows: the evader only moves when at least one 
of its neighbors is occupied by a pursuer, and in that case it moves to a 
vertex that is not a neighbor of the pursuer-occupied vertices.
We can mimic the same strategy in our polygon setting, and show that the evader
has a winning strategy against two pursuers.

The evader mimics the 
same strategy as on the original graph except now pursuers
can be located within corridors (essentially on edges in the graph setting),
as well as at vertices. However notice that if the evader confines
its movements to the vertices, and a pursuer is within a corridor
it can threaten only the vertices at each end. Where if the pursuer
was at either vertex it could still threaten both vertices and 
several more if the vertex had degree greater than one.
Thus the evader's strategy on the graph to avoid capture must still 
be viable as pursuer's gain no additional ability to threaten vertices
by being within corridors. In Figure~\ref{fig:dodecahedron} we see 
a example planar graph with no 3 or 4 cycles and all vertices of degree three, 
thus an evader can always avoid capture in the 
transformed polygon from two pursuers.

\end{proof}

\section{Closing Remarks} \label{sec:Conclusion}

In this paper, we proved that three pursuers are always sufficient to 
capture an evader in a polygonal environment of arbitrary complexity,
under the assumption that pursuers have access to evader's location at 
all times. We also proved a matching lower bound, showing that three
pursuers are also necessary in the worst-case. Traditionally, the papers 
on continuous space, visibility-based pursuit problem have focussed on 
simply detecting the evader, and not on capturing it. One of our contributions
is to isolate the \emph{intrinsic} complexity of the capture from the
associated complexity of detection or localization. In particular, while 
$\Theta (\sqrt{h} +\log n)$ pursuers are necessary (and also sufficient) 
for detection or localization of an evader in a $n$-vertex polygon with 
$h$ holes~\cite{Guibas99}, our result shows that full localization
information allows capture with only \emph{3} pursuers.
On the other hand, it still remains an intriguing open problem
whether $\Theta (\sqrt{h} +\log n)$ pursuers can \emph{simultaneously}
perform localization and capture. We leave that as a topic for future
research.

\bibliography{main}{}

\begin{thebibliography}{10}

\bibitem{aigner}
M.~Aigner and M.~Fromme.
\newblock A game of cops and robbers.
\newblock {\em Discrete Applied Mathematics}, 8(1):1--12, 1984.

\bibitem{dan-graph}
D.~Bienstock and P.~Seymour.
\newblock Monotonicity in graph searching.
\newblock {\em J. Algorithms}, 12(2):239--245, 1991.

\bibitem{cop1}
F.~V. Fomin, P.~A. Golovach, and J.~Kratochv\'{\i}l.
\newblock On tractability of cops and robbers game.
\newblock In {\em TCS}, pages 171--185, 2008.

\bibitem{Guibas99}
L.~J. Guibas, J.-C. Latombe, S.~M. LaValle, D.~Lin, and R.~Motwani.
\newblock Visibility-based pursuit-evasion in a polygonal environment.
\newblock {\em IJCGA}, 9(5):471--494, 1999.

\bibitem{Guy91}
R.~K. Guy.
\newblock Unsolved problems in combinatorial games.
\newblock In {\em Games of No Chance}, pages 475--491, 1991.

\bibitem{robot-rabbit}
B.~Halpern.
\newblock The robot and the rabbit--a pursuit problem.
\newblock {\em The American Mathematical Monthly}, 76(2):140--145, 1969.

\bibitem{suri-focs}
J.~Hershberger and S.~Suri.
\newblock "{V}ickrey pricing and shortest paths: What is an edge worth?".
\newblock In {\em 43th FOCS}, 2002.

\bibitem{Isler05}
V.~Isler, S.~Kannan, and S.~Khanna.
\newblock Randomized pursuit-evasion in a polygonal environment.
\newblock {\em Robotics, IEEE Transactions on}, 21(5):875 -- 884, 2005.

\bibitem{Isler06}
V.~Isler, S.~Kannan, and S.~Khanna.
\newblock Randomized pursuit-evasion with local visibility.
\newblock {\em SIAM Journal on Discrete Mathematics}, 1:26--41, 2006.

\bibitem{isler}
V.~Isler and N.~Karnad.
\newblock The role of information in the cop-robber game.
\newblock {\em TCS}, 399(3):179 -- 190, 2008.

\bibitem{lapaugh}
A.~LaPaugh.
\newblock Recontamination does not help to search a graph.
\newblock {\em J. ACM}, 40(2):224--245, 1993.

\bibitem{widmayer}
E.~Nardelli, G.~Proietti, and P.~Widmayer.
\newblock Finding the most vital node of a shortest path.
\newblock {\em TCS.}, 296(1):167--177, 2003.

\bibitem{Park01}
S.-M. Park, J.-H. Lee, and K.-Y Chwa.
\newblock Visibility-based pursuit-evasion in a polygonal region by a searcher.
\newblock In {\em ICALP}, pages 281--290. Springer-Verlag, 2001.

\bibitem{Par76}
T.~D. Parsons.
\newblock Pursuit-evasion in a graph.
\newblock In Y.~Alavi and D.~R. Lick, editors, {\em Theory and Application of
  Graphs}, pages 426--441. Springer-Verlag, Berlin, 1976.

\bibitem{SacRajLav04}
S.~Sachs, S.~Rajko, and S.~M. LaValle.
\newblock Visibility-based pursuit-evasion in an unknown planar environment.
\newblock {\em International Journal of Robotics Research}, 23(1):3--26, 2004.

\bibitem{sgall_lion}
J.~Sgall.
\newblock Solution of david gale's lion and man problem.
\newblock {\em Theor. Comput. Sci.}, 259(1-2):663--670, 2001.

\bibitem{Suzuki92}
I.~Suzuki and M.~Yamashita.
\newblock Searching for a mobile intruder in a polygonal region.
\newblock {\em SIAM J. Comp.}, 21:863--888, 1992.

\bibitem{Yen71}
J.~Yen.
\newblock Finding the k shortest loopless paths in a network.
\newblock {\em Management Science}, 17(11):pp. 712--716, 1971.

\end{thebibliography}
\bibliographystyle{plain}

\end{document}